\documentclass{article}

\usepackage{natbib}

 \usepackage[preprint]{neurips_2024}

\usepackage[utf8]{inputenc} 
\usepackage[T1]{fontenc}    
\usepackage{hyperref}       
\usepackage{url}            
\usepackage{booktabs}       
\usepackage{amsfonts}       
\usepackage{nicefrac}       
\usepackage{microtype}      
\usepackage{xcolor}         

\usepackage{subcaption}
\usepackage{amsmath}
\usepackage{amssymb}
\usepackage{mathtools}
\usepackage{amsthm}
\usepackage{ragged2e} 
\usepackage{booktabs,makecell, multirow, tabularx}
\usepackage[capitalize,noabbrev]{cleveref}

\usepackage{algorithm}
\usepackage{algorithmic}
\theoremstyle{plain}
\newtheorem{theorem}{Theorem}[section]

\theoremstyle{definition}

\theoremstyle{remark}

\title{Buffered Asynchronous Secure Aggregation \\ for Cross-Device Federated Learning}

\author{%
   Kun Wang, Yi-Rui Yang, Wu-Jun Li
   \thanks{Corresponding author.}
    \\
  National Key Laboratory for Novel Software Technology\\
  Department of Computer Science and Technology, Nanjing University, China\\
  \texttt{wangk@smail.nju.edu.cn, }\texttt{yangyr@smail.nju.edu.cn, }\texttt{liwujun@nju.edu.cn} \\
}

\begin{document}

\maketitle

\begin{abstract}
  Asynchronous federated learning~(AFL) is an effective method to address the challenge of device heterogeneity in cross-device federated learning. However, AFL is usually incompatible with existing secure aggregation protocols used to protect user privacy in federated learning because most existing secure aggregation protocols are based on synchronous aggregation. To address this problem, we propose a novel secure aggregation protocol named \underline{b}uffered \underline{a}synchronous \underline{s}ecure \underline{a}ggregation~(BASA) in this paper. Compared with existing protocols, BASA is fully compatible with AFL and provides secure aggregation under the condition that each user only needs one round of communication with the server without relying on any synchronous interaction among users. Based on BASA, we propose the first AFL method which achieves secure aggregation without extra requirements on hardware. We empirically demonstrate that BASA outperforms existing secure aggregation protocols for cross-device federated learning in terms of training efficiency and scalability.
\end{abstract}

\section{Introduction}
Federated learning~(FL) \citep{mcmahan2017communication} is a distributed learning paradigm that allows numerous decentralized users~(or called workers), like smartphones or edge devices, to collaboratively train a shared model while keeping training data private on their own devices. A typical FL training process involves multiple users under the coordination of a central server. During each training round, users perform local training using their private data and only share the model updates or gradients with the central server. Then, the central server collects the shared model updates from users and performs aggregation to update the global model. This mechanism eliminates the transmission and storage overhead of large datasets and prevents sensitive training data from being exposed. 

Although FL intends to protect the privacy of users' local training data, recent works have pointed out that privacy risks still exist in FL. Specifically, a malicious central server that adopts gradient leakage attack \citep{melis2019exploiting, zhu2019deep, geiping2020inverting} may be able to infer sensitive properties or even reconstruct the original training data based on user-uploaded model updates. To defend against gradient leakage attacks, the secure aggregation~(SA) protocol \citep{bonawitz2017practical} has become one of the key components of FL as an effective privacy-preserving technique. SA ensures that the server can only learn the aggregated model updates, yet can not get access to the local model update of any individual user.

However, most SA protocols are based on synchronous federated learning~(SFL) and therefore face performance issues in cross-device environments. As participant devices differ in computational power, battery life, data volume, and network bandwidth, the time each device takes to complete training and send updates to the server can vary significantly \citep{nishio2019client, Diao0T21}. The central server may spend excessive time waiting for the slowest device, leading to a reduction in the overall training efficiency. 

A recent method named FedBuff \citep{nguyen2022federated} adopts trusted execution environments~(TEEs) to combine secure aggregation with asynchronous federated learning~(AFL) \citep{xie2019asynchronous,xu2023asynchronous} to tackle the problem. However, TEEs must be supported by extra hardware and may not always be available in real applications. LightSecAgg~(LSA) \citep{so2022lightsecagg} is also proposed for AFL, but it requires users to share local masks with all other users, which might be impractical or even impossible in cross-device environments. Moreover, LSA still needs a synchronization phase to recover the aggregated model updates.

In this paper, we propose a novel SA protocol, called buffered asynchronous secure aggregation~(BASA), for AFL in cross-device environments. The contributions are summarized as follows:

\begin{itemize}  
\item {BASA is a secure aggregation protocol that is fully compatible with AFL and is practical and efficient for cross-device federated learning. }
\item {
 BASA is the first pairwise mask-based asynchronous secure aggregation protocol and is the first to adopt attribute-based encryption~(ABE) for providing asynchronous access control.
}
\item {Based on BASA, we propose the first AFL method which achieves secure aggregation without extra hardware requirements.
}
\item { Experiments show that BASA is much more efficient than existing secure aggregation protocols for cross-device federated learning.
	}
\end{itemize}

\section{Background}
\label{Background}

\paragraph{Federated Learning} In federated learning~(FL) \citep{mcmahan2017communication}, multiple data owners~(also called users or workers) collaboratively train a global machine learning model under the coordination of a central server, while keeping training data private at their own devices. The objective function of most FL problems can be formulated as follows: 	
\begin{equation}
	\min_{\mathbf{w}} \mathcal{L}(\mathbf{w})  = \sum_{i=1}^m p_i\mathcal{L}_{i}(\mathbf{w}),
\end{equation}
where $p_i \ge 0$ and $\sum_{i=1}^{m} p_i = 1$, $\mathbf{w}$ is the model parameter, $m$ is the number of users~(workers), $\mathcal{L}_i(\mathbf{w})$ is the local loss of the $i$-th user. In typical algorithms of federated learning, the server first broadcasts current model parameter $\mathbf{w}^t$ to users at each round $t$, and selects~(or waits for) a subset of users $\mathcal{U}$ to participate in training. Then, each selected~(or participating) user $i$ performs local training using its private training data, and shares the local model update $\Delta \mathbf{w}_i^{t}$ with the server. Finally, the server performs aggregation on collected model updates from $\mathcal{U}$ to update the global model. The mostly adopted method is to compute the weighted mean:
\begin{equation}
	\mathbf{w}^{t+1} = \mathbf{w}^{t} -\frac{1}{\sum_{j \in \mathcal{U}}{p_j}} \sum_{i \in \mathcal{U}}  p_i\Delta \mathbf{w}_i^{t}.
\end{equation}

According to the application scenarios, FL can be categorized into cross-silo FL and cross-device FL \citep{mammen2021federated, kairouz2021advances, zhang2020batchcrypt}.  Cross-silo FL typically involves fewer but powerful participants with substantial data and stable environments, such as organizations and institutions. Cross-device FL typically involves a vast number of individual devices with small data capacities, limited computation powers, and unreliable communications.

\paragraph{Asynchronous Federated Learning}
In large-scale cross-device environments, the devices usually have different computational power, communication bandwidth, and data volume. This often results in a ``straggler" problem: the central server may spend too much time waiting for the slowest device, leading to a severe reduction in the overall training efficiency. Asynchronous federated learning~(AFL) \citep{dutta2018slow, xie2019asynchronous, xu2023asynchronous} is more efficient than SFL in such situations. In AFL, the devices are not required to wait for a synchronized global update across all participants. The server updates the global model as soon as it receives a new model update \citep{xie2019asynchronous}, or when it receives a certain number of updates from users \citep{nguyen2022federated}. 

\paragraph{Secure Aggregation}  Secure aggregation~(SA) \citep{bonawitz2017practical, bell2020secure, so2021turbo} is an effective method to defend against privacy attacks \citep{ melis2019exploiting, zhu2019deep, geiping2020inverting}, by enabling the server to get aggregated results without knowing the local model updates of individual users. In the SA protocol, each pair of users agrees on a pairwise random seed $s_{uv}$. Meanwhile, each user $u$ samples an additional private random seed $b_{u}$. Then, user $u$ masks the local model updates as follows:
\begin{equation}
	\mathbf{y}_u = \mathbf{x}_u +  \text{PRG}(b_{u}) +  \sum_{v, u<v} \text{PRG}(s_{uv}) - \sum_{v, u>v} \text{PRG}(s_{uv}),
\end{equation}
where $\mathbf{x}_u$ is the local model update of user $u$, and $\text{PRG}$ is a pseudo random generator. The server collects all the masked updates of surviving users and recovers the required random seeds based on Shamir's secret sharing scheme \citep{shamir1979share}
. Finally, the server gets the aggregated model updates of the surviving users:
\begin{equation}
	\sum_{u\in\mathcal{U}} \mathbf{x}_u  = \sum_{u\in\mathcal{U}} (\mathbf{y}_u - \text{PRG}(b_{u})) 
	 + \sum_{u\in\mathcal{D}} (\sum_{v\in \mathcal{U}, u<v}\text{PRG}(s_{uv}) - \sum_{v\in \mathcal{U}, u>v}\text{PRG}(s_{uv})).
\end{equation}
where $\mathcal{U}$ denotes the set of surviving users, and $\mathcal{D}$ denotes the set of dropped users. 


\paragraph{Attribute-based Encryption} Attribute-based encryption~(ABE) \citep{zhang2020attribute} is a promising cryptographic primitive that enables fine-grained access control of encrypted data. ABE can be divided into two categories: ciphertext-policy attribute-based encryption~(CP-ABE) \citep{bethencourt2007ciphertext} and key-policy attribute-based encryption~(KP-ABE) \citep{goyal2006attribute}. We utilize CP-ABE in this work. In CP-ABE, the data provider specifies the access policy for the ciphertext based on attributes and does not need to know the identities of data users. Users with attributes matching the policy can access the data without direct authorization from the data provider. An attribute authority~(AA) manages the attributes and the corresponding secret keys. There are four main algorithms in a CP-ABE scheme:
\begin{itemize}
	\item \verb|Setup|$()$: run by the AA at the beginning, generates the public key $\mathrm{PK}$ and master key $\mathrm{MK}$ for the system.
	
	\item \verb|KeyGen|$(\mathrm{PK}, \mathrm{MK}, \mathcal{A})$: run by the AA, generates the attribute secret key $\mathrm{SK}$ for attribute set $\mathcal{A}$.
	
	\item \verb|Encrypt|$(\mathrm{PK}, M, \mathbb{A})$: run by the data owner, encrypts message $M$ under access policy $\mathbb{A}$, and outputs a ciphertext $\mathrm{CT}_{\mathbb{A}}$.
	
	\item \verb|Decrypt|$(\mathrm{PK}, \mathrm{CT}_{\mathbb{A}}, \mathrm{SK})$: run by the data user, decrypts $\mathrm{CT}_{\mathbb{A}}$ if $\mathrm{SK}$ matches the corresponding access policy $\mathbb{A}$, and return $\perp$ otherwise.
\end{itemize}

\section{Method}
\label{method}
In this section, we present the details of our novel secure aggregation protocol named buffered asynchronous secure aggregation~(BASA).
\begin{figure}[t]
	\vskip 0.2in
	\begin{center}
		\centerline{\includegraphics[width=0.7\linewidth]{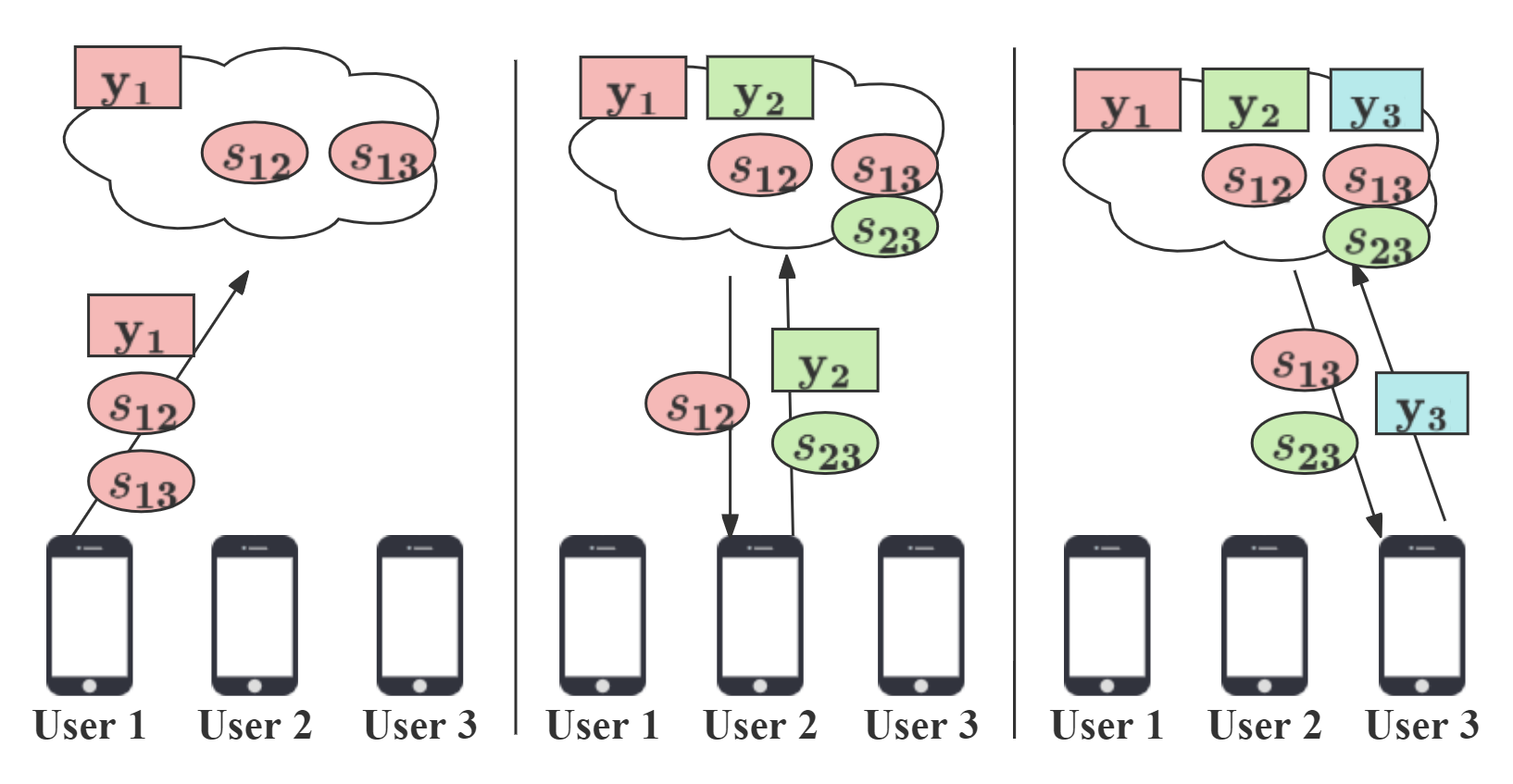}}
		\caption{A brief illustration of the pairwise mask-based framework of BASA with a simple case study of three users. Each user uploads local inputs in turn from left to right. $\mathbf{y}_1, \mathbf{y}_2, \mathbf{y}_3$ are the masked inputs, and $s_{12}, s_{13}, s_{23}$ are encrypted ciphertexts of the corresponding random seeds.}
		\label{illustration-BASA}
	\end{center}
	\vskip -0.2in
\end{figure}
\subsection{Intuition}
Secure aggregation cannot be directly integrated with AFL due to the following two main challenges:
\begin{itemize}
\item Secure aggregation typically requires a fixed cohort of users to participate in each training round. Individual users' local model updates are mixed with those of other users in the cohort to ensure privacy. However, users independently send model updates in AFL, and the server cannot predetermine which clients will be involved in aggregation. 
\item Some techniques in secure aggregation (e.g., secret sharing) rely on synchronous communication where users need to collaborate and share information synchronously. This contradicts with AFL in which users are not necessarily online at the same time and they communicate with the server asynchronously. 
\end{itemize}

Our basic idea is to adopt a buffered aggregation scheme. More specifically, during the training process, the server maintains a secure buffer of size $K$ and temporarily stores new local model updates in the buffer as soon as they are received. Once the buffer accumulates $K$ model updates, the server proceeds to aggregate all the model updates stored in the buffer and updates the global model. The secure buffer should ensure that the server can only learn the aggregated updates but cannot learn any individual user update in this buffer.

\paragraph{Pairwise Mask-based Framework} To implement the secure buffer, we develop a pairwise mask-based framework. We observe that due to device heterogeneity, users tend to complete training at different times and upload local model updates sequentially rather than uploading them at the same time. Hence, a user who uploads beforehand can generate the pairwise mask locally and upload the corresponding random seeds to the server along with the masked local model updates. In addition to buffering the model updates, the server temporarily stores the random seeds of masks. Then, the subsequent user can download that random seed on arrival and put it to the local model updates so that all masks can be canceled out when aggregated. 

We use a simple case study to illustrate this in Figure \ref{illustration-BASA}, in which we have a server with buffer size 3 and three users $U_1, U_2, U_3$ upload their local inputs $\mathbf{x}_1, \mathbf{x}_2, \mathbf{x}_3$ sequentially. We assume every input $\mathbf{x}_i$ has been quantized by users so that $\mathbf{x}_i \in \mathbb{F}_{q}^d$, where $d$ is the model dimension and $q$ is the size of the finite field. 

First, $U_1$ locally generates two random seeds $s_{12}$ and $s_{13}$, and masks $\mathbf{x}_1$ as $$\mathbf{y}_1 = \mathbf{x}_1 + \text{PRG}(s_{12}) + \text{PRG}(s_{13}).$$ 
Then, $U_1$ uploads $\mathbf{y}_1$, $s_{12}$ and $s_{13}$ to the server. When $U_2$ comes, it downloads $s_{12}$ from the server and locally generates a random seed $s_{23}$. $U_2$ masks $\mathbf{x}_2$ as 
$$\mathbf{y}_2 = \mathbf{x}_2 - \text{PRG}(s_{12}) + \text{PRG}(s_{23}).$$
Similarly, $\mathbf{y}_2$ and $s_{13}$ are also sent to the server. Finally, $U_3$ downloads $s_{13}$ and $s_{23}$, and masks the local input $\mathbf{x}_3$ as $$\mathbf{y}_3 = \mathbf{x}_3 - \text{PRG}(s_{13}) - \text{PRG}(s_{23}).$$ $U_3$ only needs to send $\mathbf{y}_3$ to the server. After receiving $\mathbf{y}_3$, the server performs aggregation over $\mathbf{y}_1, \mathbf{y}_2$ and $\mathbf{y}_3$. The above method ensures that $\sum_{i=1}^3 \mathbf{y}_i = \sum_{i=1}^3 \mathbf{x}_i$, as all random masks can be canceled out when aggregation.

\paragraph{ABE-based Random Seed Sharing} To protect user privacy, the users must upload encrypted random seeds to prevent the server from knowing the pairwise masks. However, the identity of the decrypting user cannot be predetermined at the time of encryption because the subsequent users have not yet connected.

We tackle this problem by incorporating attribute-based encryption~(ABE). More specifically, we assign attributes to positions within the buffer. During encryption, users can specify the access policy for the ciphertext. When subsequent users connect with the server, they can obtain the matching attributes based on their location in the buffer, and decrypt the corresponding ciphertexts. For example, $U_1$ can specify that $s_{12}$ can only be decrypted by users with the attribute ``Second". When $U_2$ uploads its model update, it is the second user in the buffer and gets the attribute ``Second". Therefore, $U_2$ can decrypt $s_{12}$ without authorization from $U_1$. In this way, each user needs only a single round of communication with the server and does not have to be online all the time.


\subsection{Detailed Protocol}

\begin{algorithm}[t]
	\caption{BASA-Server}
	\label{alg:basa}
	\begin{algorithmic}[1]
		\STATE {\bfseries Input:} buffer size $K$.
		\STATE {\bfseries Output:} the aggregated result $\bar{\mathbf{y}}$.
		\STATE Initialize $k \leftarrow 0, \bar{\mathbf{y}} \leftarrow 0, \mathcal{A} \leftarrow \{a_k \mid k<K\} $;
		\FOR{$ i = 0 $ to $ K$}
		\STATE $\mathcal{C}_i \leftarrow \emptyset$;
		\ENDFOR
		\LOOP
		\IF{$\text{receive connection from user $u$ }$}
		\STATE Send $\mathcal{A}$, $k$, $\mathcal{C}_k$ to user $u$;
            \IF{timeout}
                \STATE Abort current connection;
            \ENDIF
		\STATE $\mathbf{y}_k \leftarrow $ masked inputs from user $u$;
		\STATE $\{c_{kj} \mid k<j<K\} \leftarrow $  ciphertexts from user $u$;
		\FOR{$ j = k+1 $ to $ K$}
		\STATE $\mathcal{C}_j \leftarrow \mathcal{C}_j \cup \{c_{kj}\} $;
		\ENDFOR 
		\STATE $\bar{\mathbf{y}} \leftarrow \bar{\mathbf{y}} + \mathbf{y}_k$ \STATE $k \leftarrow k+1$;
		\ENDIF
		\IF{$k ==K$}
		\STATE Output $\bar{\mathbf{y}}$ as the aggregated result;
		\ENDIF
		\ENDLOOP
	\end{algorithmic}
\end{algorithm}

\begin{algorithm}[t]
	\caption{BASA-User}
	\label{alg:basa_user}
	\begin{algorithmic}[1]
		\STATE {\bfseries Input:} user input $\mathbf{x}_k$, buffer size $K$, position in buffer $k$, attributes set $\mathcal{A} = \{a_k \mid k<K\}$, ciphertext set $\mathcal{C}_k$.
		\STATE Initialize $\mathbf{y}_k \leftarrow \mathbf{x}_k;$ 
		\STATE Request attribute secret key $\mathrm{SK}_{k}$ from AA;
		\FOR {$ i = 1 $ to $ k$}
		\STATE $s_{ik} \leftarrow \mathrm{Dec}(c_{ik}, \mathrm{SK}_{k})$; 
		\STATE $\mathbf{y}_k  \leftarrow \mathbf{y}_k - \text{PRG}(s_{ik})$;
		\ENDFOR
		\FOR {$ j = k+1 $ to $ K$}
		\STATE $s_{kj} \leftarrow $ generated random seed;
		\STATE $\mathbf{y}_k  \leftarrow \mathbf{y}_k+ \text{PRG}(s_{kj})$;
		\STATE  $c_{kj} = \mathrm{Enc}(s_{kj}, a_j)$;
		\ENDFOR
		\STATE Send $\mathbf{y}_k$ to server;
		\STATE Send $\{c_{kj} \mid k<j<K\} $ to server;
	\end{algorithmic}
\end{algorithm}

The protocol of BASA is detailed in Algorithm \ref{alg:basa} and  Algorithm \ref{alg:basa_user}. It is composed of a server, $N$ users, and an attribute authority~(AA). The AA is a trusted party for ABE that publishes the system's public keys, and issues attribute secret keys for ABE. The server maintains a secure buffer of size $K$. The users independently choose to participate in aggregation based on their availability. Assume every update of users has been quantized in $\mathbb{F}_{q}^d$, where $d$ is the model dimension and $q$ is the finite field size.

Initially, the server associates each position $k$ of secure buffer with an attribute $a_k$ and assigns a ciphertext buffer $\mathcal{C}_k$ for each position. The AA runs the system setup algorithm of ABE. When a new user wishes to upload a model update, it first connects to the server. Suppose the user $u$ is the $k$-th user to connect. The server sends to $u$ the attributes set $\mathcal{A} = \{a_i \mid i<K\}$, the position $k$ of $u$ in the secure buffer, and all ciphertext $c_{ik}$ in $\mathcal{C}_k$, for $0< i < k$.

After connection, user $u$ requires the attribute secret key $\mathrm{SK}_{k}$ corresponding to its position attribute $a_k$ from the AA. Then $u$ can use $\mathrm{SK}_{k}$ to decrypt all $c_{ik}$, for $0< i < k$.
User $u$ also generates random seed $s_{kj}$ and encrypts it under the attribute $a_j$ for all $k < j < K$.
Finally, user $u$ masks its local input $\mathbf{x}_k$ as:
\begin{equation}
	\mathbf{y}_k =  \mathbf{x}_k- \sum_{0< i < k} \text{PRG}(s_{ik}) + \sum_{k < j < K} \text{PRG}(s_{kj}),
\end{equation}
and sends $\mathbf{y}_k$ and all $c_{kj}$ for $k < j < K$ to the server. 

The server stores the masked inputs in the buffer and puts each $c_{kj}$ in $\mathcal{C}_j$. Upon receiving $K$ masked inputs from users, the server conducts aggregation $\bar{\mathbf{y}} = \sum_{i<K}\mathbf{y}_i$. Furthermore, the server needs to regenerate the attribute set $\mathcal{A}$ to avoid the user still using the attribute secret key in the subsequent aggregation process. If user $u$ drops out after connecting to the server, the server can wait for a timeout time $t_{out}$, then aborts the connection with $u$ and continues to listen to subsequent connections.

\subsection{AFL with BASA}

The BASA protocol can be directly employed in AFL. More specifically, the server sends the global model to users asynchronously and runs \verb|BASA-Server| to receive model updates uploaded by users. Upon receiving $K$ masked updates, the server can update the global model using the aggregated results of BASA. The detailed training algorithm on the server is presented in Appendix \ref{BASA-AFL}.

Each user receives the global model from the server and performs local training. Specifically, each user updates the model parameter with multiple SGD steps using its local dataset, as in \mbox{FedBuff}~\citep{nguyen2022federated}. After local training, the user expands the local model update $\Delta_k$ to a one-dimension vector and quantizes it to the finite domain $\mathbb{F}_{q}^d$, where $d$ is the model dimension and $q$ is the size of the finite field. We use the same stochastic quantization strategy as LSA \citep{so2022lightsecagg}. Finally, the user takes it as the input of \verb|BASA-User| to upload to the server. 

AFL with BASA can also consider the effect of model staleness during aggregation. Specifically, when a user completes local training and starts a connection, the server sends the user the current global model timestamp $t$. The user then calculates the staleness factor $\alpha_k = S(t - \tau_k)$ from the received global timestamp $t$ and the local timestamp $\tau_k$ and masks the local model update as:
\begin{equation}
	\tilde\Delta_k = \alpha_k \Delta_k - \sum_{0< i < k} \text{PRG}(s_{ik}) + \sum_{k < j < K} \text{PRG}(s_{kj}),
\end{equation}
$\tilde\Delta_k$ and $\alpha_k$ are sent to the server, and the aggregation is modified as
$\Bar{\Delta} = {\sum_{k<K} \tilde\Delta_k} / {\sum_{k<K}\alpha_k}$.

\subsection{Security Analysis}
\label{Security}

\paragraph{Threat Model} We assume the server and users are honest but curious. They will correctly follow the protocol but are interested to learn the private information of other users as much as possible. Moreover, some users may collude with the server by sharing their inputs and the associated random masks. Furthermore, we assume the AA is a trusted third party. 

\paragraph{Security Guarantee} We argue that BASA can provide security under the honest-but-curious threat model. Specifically, when the buffer size is $K$, BASA is able to protect the honest user's local model updates in the presence of at most $K-2$ colluded users. Let $\mathcal{U}$ be the set of participating users in a round of the aggregation, $\mathcal{S} \subseteq \mathcal{U}$ be the set of colluded users in $\mathcal{U}$. We have the following result:
\begin{theorem}
	\label{thm:security}
	  In BASA, the server can only learn the partial aggregation result over the inputs of honest users $\sum_{i \in \mathcal{U} \setminus \mathcal{S}} \mathbf{x}_i$, under the honest-but-curious threat model.
\end{theorem}
\begin{proof} 
	Consider in the execution of BASA, the server has accepted $k$ updates from users. At this point, the partial aggregation result is:
	\begin{equation}
	\label{eq:security_1}
		\sum_{i \leq k} \tilde\Delta_i = \sum_{i \leq k} \Delta_i + \sum_{i \leq k} \left(\sum_{k<j<K} \text{PRG}\left(s_{ij}\right) \right).
	\end{equation}
	Under the honest-but-curious threat model, the $\Delta_i$ for all $ i \in \mathcal{S}$, and the random seeds $s_{ij}$ for $i \in \mathcal{S}$ or $j \in \mathcal{S}$, are shared with the server by colluded users. But the random seeds $s_{ij}$ where $i,j \in \mathcal{U} \setminus \mathcal{S}$ are not revealed to the server and other users, due to the access control provided by the ABE scheme. Therefore, the server can remove all known entries in (\ref{eq:security_1}) and obtains:
	\begin{equation}
		\label{eq:security_2}
		\sum_{i \in \mathcal{U} \setminus \mathcal{S}, i \leq k} \Delta_i + \sum_{i \in \mathcal{U} \setminus \mathcal{S}, i \leq k}  \left(\sum_{j \in \mathcal{U} \setminus \mathcal{S}, k<j<K} \text{PRG}\left(s_{ij}\right) \right).
	\end{equation}
	If there exists a user $j \in \mathcal{U} \setminus \mathcal{S}$ and $j > k$, the partical result obtained by server looks uniformly random due to the presence of $\text{PRG}\left(s_{ij}\right)$. If every user $j > k$ colludes with the server, the latter item in (\ref{eq:security_2}) is zero and only $\sum_{i \in \mathcal{U} \setminus \mathcal{S}, i \leq k} \Delta_i$ is left. This holds for all $k < K$.
\end{proof}
Therefore, as long as there exists more than one honest user in the buffer, the local model updates of individual honest users are not revealed to the server. In other words, the honest user's local model updates are protected in the presence of at most $K-2$ colluded users.

\section{Experiments}
\label{experiments}

In this section, we empirically evaluate our proposed BASA protocol and other baselines in cross-device federated learning.

\subsection{Setup}
\label{experiments-setup}
\paragraph{Baselines} We compare BASA with three baseline protocols: SecAgg~\citep{bonawitz2017practical} in SFL, LightSecAgg~(LSA) in SFL~\citep{so2022lightsecagg}, and LightSecAgg~(LSA) in AFL~\citep{so2021secure}. Furthermore, we compare BASA with FedBuff \citep{nguyen2022federated} without secure aggregation to evaluate the additional time cost introduced by secure aggregation. For SFL, we adopt FedAvg \citep{mcmahan2017communication} as the basic training algorithm. For AFL, we adopt the FedBuff as the basic training algorithm. We don't compare with FedBuff under TEE because we only consider cross-device environments without such extra hardware support.

\paragraph{Tasks} We evaluate these protocols for federated learning tasks on three datasets: MNIST, \mbox{FEMNIST}~\citep{caldas2018leaf}, and CIFAR-10~\citep{krizhevsky2009learning}. For the MNIST dataset, the data is divided across 1000 devices in a non-i.i.d. way. We train a simple logistic regression model. For the FEMNIST, we adopt the same data partition as in LEAF \citep{caldas2018leaf} for 3400 users. For the CIFAR-10 dataset, we partition the data among 100 users. We train CNN \citep{mcmahan2017communication} models for FEMNIST and CIFAR-10.  

We implement the training experiments in the FedML framework~\citep{he2020fedml}. For large-scale simulation experiments, we use the FLSim framework that is used in FedBuff~\citep{nguyen2022federated}. More details about the experiments are provided in Appendix \ref{appendix-Experiment Details}.


\subsection{Performance Evaluation and Analysis}
\begin{table}[t]
	\caption{The wall-clock training time (in second) to reach a target validation accuracy for 32 users when training logistic regression on MNIST and CNN models on FEMNIST and CIFAR-10. A larger $\beta$ implies a more severe straggler problem. }
	\centering
         \resizebox{1.0\linewidth}{!}{
\begin{tabular}{|c|c|c|c|c|c|c|c|}
\hline
Datasets                  & Accuracy              & \multicolumn{1}{l|}{Delay Scale} & \begin{tabular}[c]{@{}c@{}}No SA\\ (AFL)\end{tabular} & \begin{tabular}[c]{@{}c@{}}BASA\\ (AFL)\end{tabular} & \begin{tabular}[c]{@{}c@{}}SecAgg\\ (SFL)\end{tabular} & \begin{tabular}[c]{@{}c@{}}LSA\\ (SFL)\end{tabular} & \begin{tabular}[c]{@{}c@{}}LSA\\ (AFL)\end{tabular} \\ \hline
\multirow{3}{*}{MNIST}    & \multirow{3}{*}{80\%} & $\beta=0$                        & 9.97$\pm{0.18}$                                       & 22.62$\pm{0.23}$                                     & 7.85$\pm{0.31}$                                        & 7.73$\pm{1.02}$                                     & 16.95$\pm{0.35}$                                    \\ \cline{3-8} 
                          &                       & $\beta=3.0$                      & 36.66$\pm{1.56}$                                      & 40.17$\pm{2.71}$                                     & 129.04$\pm{10.12}$                                     & 128.93$\pm{9.29}$                                   & 116.58$\pm{6.51}$                                   \\ \cline{3-8} 
                          &                       & $\beta=6.0$                      & 46.61$\pm{1.93}$                                      & 57.01$\pm{2.24}$                                     & 164.52$\pm{23.69}$                                     & 152.12$\pm{12.26}$                                  & 149.25$\pm{8.81}$                                   \\ \hline
\multirow{3}{*}{FEMNIST}  & \multirow{3}{*}{75\%} & $\beta=0$                        & 54.63$\pm{0.95}$                                      & 248.51$\pm{5.05}$                                    & 366.45$\pm{14.62}$                                     & 199.88$\pm{10.13}$                                  & 361.85$\pm{4.84}$                                   \\ \cline{3-8} 
                          &                       & $\beta=3.0$                      & 154.08$\pm{2.16}$                                     & 251.09$\pm{3.92}$                                    & 849.12$\pm{20.02}$                                     & 691.20$\pm{8.05}$                                   & 1005.86$\pm{27.13}$                                 \\ \cline{3-8} 
                          &                       & $\beta=6.0$                      & 421.51$\pm{1.92}$                                     & 475.89$\pm{7.41}$                                    & 1434.01$\pm{31.14}$                                    & 1229.09$\pm{33.56}$                                 & 2055.27$\pm{37.85}$                                 \\ \hline
\multirow{3}{*}{CIFAR-10} & \multirow{3}{*}{50\%} & $\beta=0$                        & 85.34$\pm{0.43}$                                      & 107.39$\pm{1.02}$                                    & 158.58$\pm{6.93}$                                      & 90.67$\pm{3.20}$                                    & 326.56$\pm{1.41}$                                   \\ \cline{3-8} 
                          &                       & $\beta=3.0$                      & 250.98$\pm{2.33}$                                     & 256.52$\pm{2.26}$                                    & 737.05$\pm{21.73}$                                     & 683.15$\pm{19.41}$                                  & 1166.71$\pm{15.74}$                                 \\ \cline{3-8} 
                          &                       & $\beta=6.0$                      & 477.96$\pm{8.34}$                                     & 479.03$\pm{5.02}$                                    & 1248.39$\pm{29.84}$                                    & 1143.12$\pm{56.92}$                                 & 2280.76$\pm{43.79}$                                 \\ \hline
\end{tabular}
	}
	\label{tab:performance}
\end{table}

For performance evaluation, we perform the training process with one server and 32 concurrent users. For SFL, the server randomly selects 32 users per round to participate in training. For AFL, we always keep 32 users training concurrently, and the server's buffer size is set to 10. In addition, we do not consider user dropouts in our experiments. To simulate the device heterogeneity, we introduce a random delay drawn from an exponential distribution to the local training time of users. The scale parameter $\beta=1/\lambda$ of the distribution controls the severity of the straggler problem, where $\lambda$ is the rate parameter of the exponential distribution. A larger $\beta$ causes a more severe straggler problem. We run the experiments with $\beta=0$, $\beta=3.0$ and $\beta=6.0$ respectively. 

We measure the required wall-clock time for each protocol to reach the target validation accuracy on each training task.  The result is shown in Table \ref{tab:performance}. We can find that when $\beta$ is not zero, which indicates the presence of stragglers, BASA takes the shortest total wall-clock time to reach the target accuracy compared with other secure aggregation protocols, achieving up to 4.7x speedup when $\beta=6$ and up to 4.5x speedup when $\beta=3$. This highlights BASA's efficiency in the presence of straggler problems while simultaneously ensuring secure aggregation.

Compared to FedBuff~(denoted as ``No SA'' in Table~\ref{tab:performance}) in which no secure aggregation is employed, BASA incurs up to 23\% extra time cost when $\beta=6$ and up to 62\% extra time cost when $\beta=3$. This suggests that the extra relative overhead introduced by BASA becomes smaller as the straggler problem becomes more severe. When $\beta = 0$, which indicates minimal or no straggler problems, BASA tends to be less efficient than secure aggregation of SFL. This is because BASA requires users to execute the protocol serially when submitting their model updates. Therefore, when multiple users complete training simultaneously, only one user can submit updates to the server at a time, while the rest must wait. One solution for this is to maintain multiple buffers on the server and run multiple BASA instances simultaneously, thus increasing the efficiency of the server. But this is not the focus of this paper, and is left for future work.



\subsection{Large-scale Training Simulation}

We simulate the large-scale training process using the simulator framework proposed in \mbox{FedBuff}~\citep{nguyen2022federated}. We perform the training tasks on the FEMNIST dataset and CIFAR-10 dataset with 200 concurrent users. We still adopt the exponentially distributed delay to simulate the device heterogeneity. The scale of the delay distribution varies from 1 to 6. We evaluate the total wall-clock time to reach a target accuracy for AFL with BASA and SFL with SecAgg. 

As shown in Figure \ref{simulation-BASA}, the results illustrate that AFL with BASA can significantly outperform SFL with SecAgg when the divergence~(heterogeneity) of user training time increases. This is mainly because BASA reduces the waiting time for ``straggler'' users. Moreover, AFL with BASA updates the global model parameters more frequently than SFL, as the buffer size is much smaller than the number of concurrent training users. We also perform the training tasks with different numbers of concurrent users, the results of which are shown in Appendix \ref{Additional_Experiments}.

\begin{figure}[t]
	\centering
             \subcaptionbox{CIFAR-10}{\includegraphics[width=0.45\linewidth]{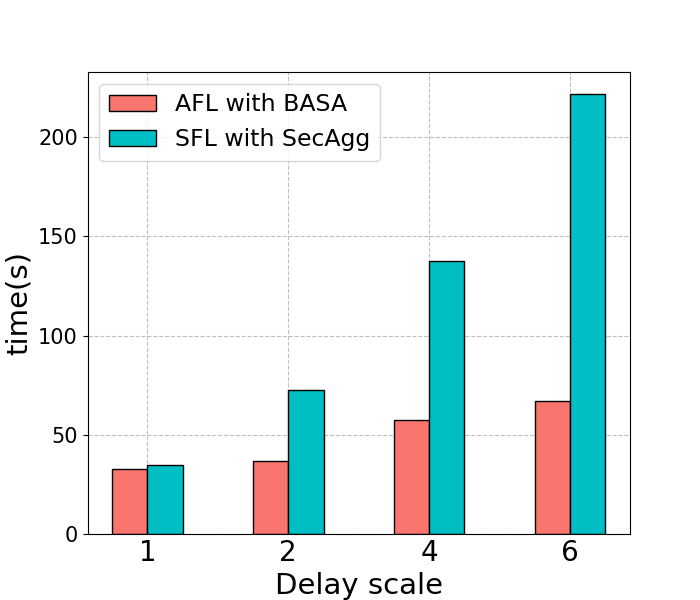}}
             \subcaptionbox{FEMNIST}{\includegraphics[width=0.45\linewidth]{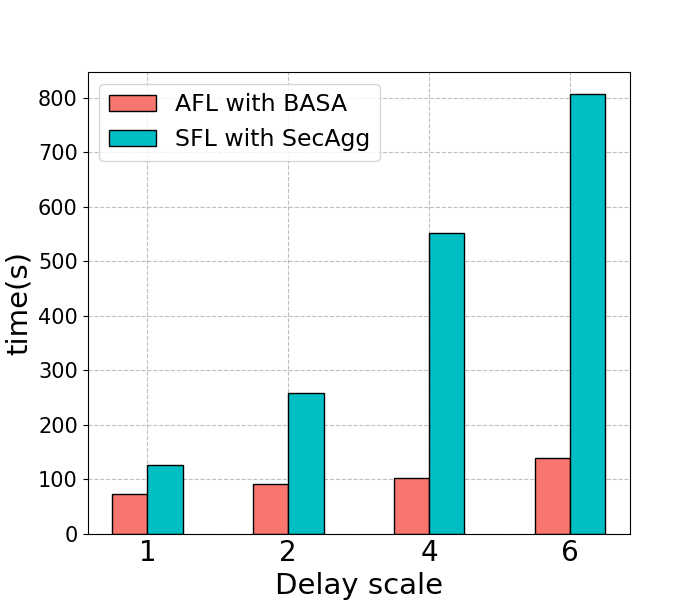}}
             
     \caption{The wall-clock time (in second)  to reach a target validation accuracy for 200 users when training CNN models on the FEMNIST and CIFAR-10 datasets with different delay scales. }
    \label{simulation-BASA}
\end{figure}


\subsection{Scalability Evaluation in AFL} 
\begin{figure}[t]
	\centering
             \subcaptionbox{}{\includegraphics[width=0.46\linewidth]{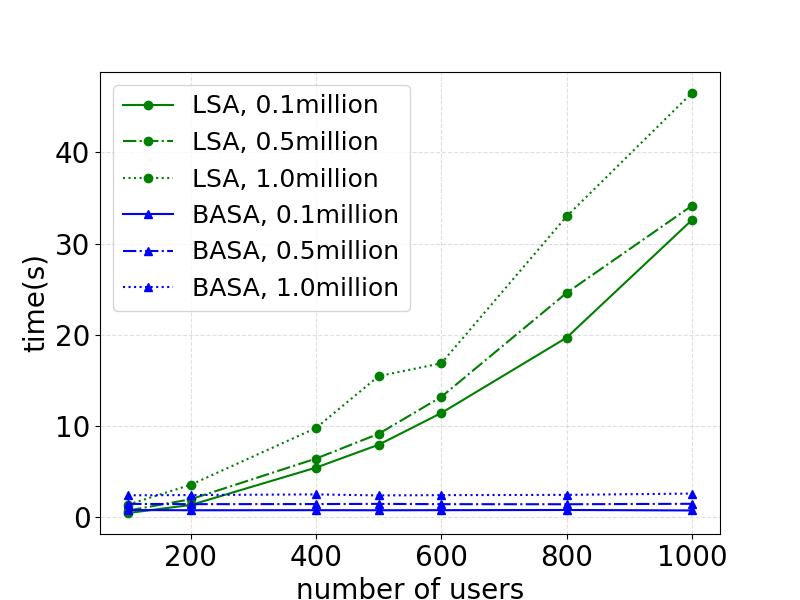}}
             \subcaptionbox{}{\includegraphics[width=0.46\linewidth]{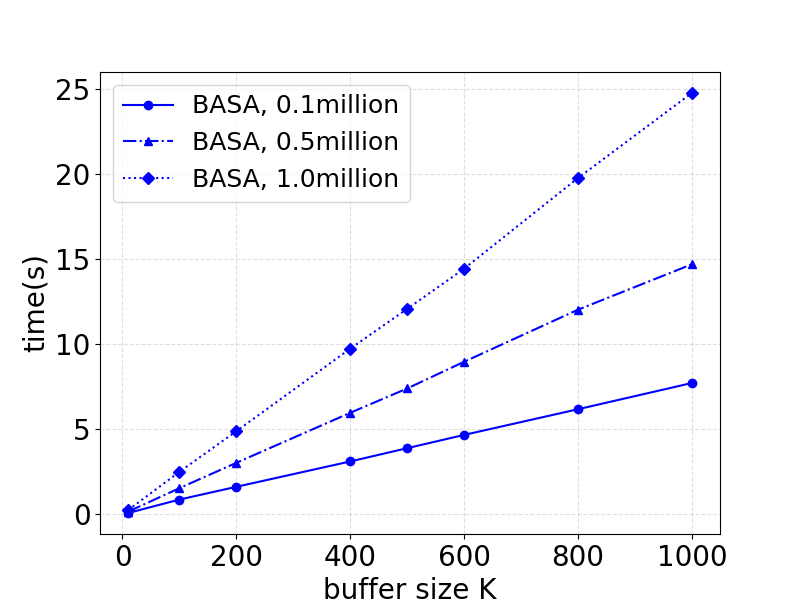}}
             
     \caption{The wall-clock time~(in second) taken by a user to execute the BASA protocol and LSA protocol with different user numbers and buffer sizes. The terms ``$0.1$million'', ``$0.5$million'', and ``$1.0$million'' refer to different model dimensions. }
    \label{scalability-BASA}
\end{figure}
To evaluate scalability, we compare the additional computation cost of BASA and LSA in AFL \citep{so2021secure} as the number of participating users and the model dimension increase. We test the model dimensions of $0.1$million, $0.5$million, and $1.0$million, respectively. The total number of users varies from 100 to 1000. Figure \ref{scalability-BASA}(a) illustrates the wall-clock time a user takes to execute the protocol once with a varying number of users and model dimensions. With the increasing number of users and model dimensions, the time required to execute LSA also experiences a considerable increase. This is mainly due to the additional time required for computing and sharing local masks among all users during the offline phase of LSA. On the contrary, BASA only requires sharing random mask seeds among users within the buffer. Hence, the time required to execute BASA only slightly increases when increasing the number of users and model dimensions.

 In Figure \ref{scalability-BASA}(b), we illustrate the wall-clock time for a user to execute BASA with buffer size $K$ varying from 10 to 1000. It can be seen that the execution time increases linearly with the buffer size, which is mainly caused by the ciphertext encryption/decryption in ABE and the computation of random masks. As BASA adopts a serial execution mode among users, the total time required to complete one aggregation process exhibits a quadratic increase to $K$. However, $K$ is often fixed to a small value in asynchronous federated learning, such as $K=10$ in FedBuff \citep{nguyen2022federated}. Therefore, BASA can still maintain high efficiency in large-scale training.


\section{Conclusion}
\label{conclusion}
In this paper, we introduce buffered asynchronous secure aggregation~(BASA), a novel secure aggregation protocol designed for compatibility with asynchronous federated learning. Compared with existing protocols, BASA improves the efficiency of secure aggregation in cross-device federated learning. Based on BASA, we propose the first AFL method which achieves secure aggregation without extra requirements on hardware. The empirical results show that our method not only achieves substantial speedups in scenarios with straggler problems but also exhibits remarkable scalability. This makes BASA notably practical for large-scale cross-device federated learning scenarios, where device heterogeneity is a common challenge. Our proposed method is limited to the honest-but-curious threat model. Extending BASA to other threat models, such as the malicious server model, is also interesting and will be pursued in future work.




\bibliography{neurips_2024}

\begin{thebibliography}{26}
\providecommand{\natexlab}[1]{#1}
\providecommand{\url}[1]{\texttt{#1}}
\expandafter\ifx\csname urlstyle\endcsname\relax
  \providecommand{\doi}[1]{doi: #1}\else
  \providecommand{\doi}{doi: \begingroup \urlstyle{rm}\Url}\fi

\bibitem[Bell et~al.(2020)Bell, Bonawitz, Gasc{\'o}n, Lepoint, and Raykova]{bell2020secure}
James~Henry Bell, Kallista~A Bonawitz, Adri{\`a} Gasc{\'o}n, Tancr{\`e}de Lepoint, and Mariana Raykova.
\newblock Secure single-server aggregation with (poly) logarithmic overhead.
\newblock In \emph{ACM SIGSAC Conference on Computer and Communications Security}, 2020.

\bibitem[Bethencourt et~al.(2007)Bethencourt, Sahai, and Waters]{bethencourt2007ciphertext}
John Bethencourt, Amit Sahai, and Brent Waters.
\newblock Ciphertext-policy attribute-based encryption.
\newblock In \emph{IEEE Symposium on Security and Privacy}, 2007.

\bibitem[Bonawitz et~al.(2017)Bonawitz, Ivanov, Kreuter, Marcedone, McMahan, Patel, Ramage, Segal, and Seth]{bonawitz2017practical}
Keith Bonawitz, Vladimir Ivanov, Ben Kreuter, Antonio Marcedone, H~Brendan McMahan, Sarvar Patel, Daniel Ramage, Aaron Segal, and Karn Seth.
\newblock Practical secure aggregation for privacy-preserving machine learning.
\newblock In \emph{ACM SIGSAC Conference on Computer and Communications Security}, 2017.

\bibitem[Caldas et~al.(2018)Caldas, Duddu, Wu, Li, Kone{\v{c}}n{\`y}, McMahan, Smith, and Talwalkar]{caldas2018leaf}
Sebastian Caldas, Sai Meher~Karthik Duddu, Peter Wu, Tian Li, Jakub Kone{\v{c}}n{\`y}, H~Brendan McMahan, Virginia Smith, and Ameet Talwalkar.
\newblock Leaf: A benchmark for federated settings.
\newblock \emph{CoRR}, abs/1812.01097, 2018.

\bibitem[Diao et~al.(2021)Diao, Ding, and Tarokh]{Diao0T21}
Enmao Diao, Jie Ding, and Vahid Tarokh.
\newblock Heterofl: Computation and communication efficient federated learning for heterogeneous clients.
\newblock In \emph{International Conference on Learning Representations}, 2021.

\bibitem[Dutta et~al.(2018)Dutta, Joshi, Ghosh, Dube, and Nagpurkar]{dutta2018slow}
Sanghamitra Dutta, Gauri Joshi, Soumyadip Ghosh, Parijat Dube, and Priya Nagpurkar.
\newblock Slow and stale gradients can win the race: Error-runtime trade-offs in distributed sgd.
\newblock In \emph{International Conference on Artificial Intelligence and Statistics}, 2018.

\bibitem[Geiping et~al.(2020)Geiping, Bauermeister, Dr{\"o}ge, and Moeller]{geiping2020inverting}
Jonas Geiping, Hartmut Bauermeister, Hannah Dr{\"o}ge, and Michael Moeller.
\newblock Inverting gradients-how easy is it to break privacy in federated learning?
\newblock In \emph{Advances in Neural Information Processing Systems}, 2020.

\bibitem[Goyal et~al.(2006)Goyal, Pandey, Sahai, and Waters]{goyal2006attribute}
Vipul Goyal, Omkant Pandey, Amit Sahai, and Brent Waters.
\newblock Attribute-based encryption for fine-grained access control of encrypted data.
\newblock In \emph{ACM Conference on Computer and Communications Security}, 2006.

\bibitem[He et~al.(2020)He, Li, So, Zeng, Zhang, Wang, Wang, Vepakomma, Singh, Qiu, et~al.]{he2020fedml}
Chaoyang He, Songze Li, Jinhyun So, Xiao Zeng, Mi~Zhang, Hongyi Wang, Xiaoyang Wang, Praneeth Vepakomma, Abhishek Singh, Hang Qiu, et~al.
\newblock Fedml: A research library and benchmark for federated machine learning.
\newblock \emph{arXiv preprint arXiv:2007.13518}, 2020.

\bibitem[Kairouz et~al.(2021)Kairouz, McMahan, Avent, Bellet, Bennis, Bhagoji, Bonawitz, Charles, Cormode, Cummings, et~al.]{kairouz2021advances}
Peter Kairouz, H~Brendan McMahan, Brendan Avent, Aur{\'e}lien Bellet, Mehdi Bennis, Arjun~Nitin Bhagoji, Kallista Bonawitz, Zachary Charles, Graham Cormode, Rachel Cummings, et~al.
\newblock Advances and open problems in federated learning.
\newblock \emph{Foundations and Trends{\textregistered} in Machine Learning}, 2021.

\bibitem[Krizhevsky et~al.(2009)Krizhevsky, Hinton, et~al.]{krizhevsky2009learning}
Alex Krizhevsky, Geoffrey Hinton, et~al.
\newblock Learning multiple layers of features from tiny images.
\newblock 2009.

\bibitem[Mammen(2021)]{mammen2021federated}
Priyanka~Mary Mammen.
\newblock Federated learning: Opportunities and challenges.
\newblock \emph{arXiv preprint arXiv:2101.05428}, 2021.

\bibitem[McMahan et~al.(2017)McMahan, Moore, Ramage, Hampson, and y~Arcas]{mcmahan2017communication}
Brendan McMahan, Eider Moore, Daniel Ramage, Seth Hampson, and Blaise~Aguera y~Arcas.
\newblock Communication-efficient learning of deep networks from decentralized data.
\newblock In \emph{Artificial Intelligence and Statistics}, 2017.

\bibitem[Melis et~al.(2019)Melis, Song, De~Cristofaro, and Shmatikov]{melis2019exploiting}
Luca Melis, Congzheng Song, Emiliano De~Cristofaro, and Vitaly Shmatikov.
\newblock Exploiting unintended feature leakage in collaborative learning.
\newblock In \emph{{IEEE} Symposium on Security and Privacy}, 2019.

\bibitem[Nguyen et~al.(2022)Nguyen, Malik, Zhan, Yousefpour, Rabbat, Malek, and Huba]{nguyen2022federated}
John Nguyen, Kshitiz Malik, Hongyuan Zhan, Ashkan Yousefpour, Mike Rabbat, Mani Malek, and Dzmitry Huba.
\newblock Federated learning with buffered asynchronous aggregation.
\newblock In \emph{International Conference on Artificial Intelligence and Statistics}, 2022.

\bibitem[Nishio and Yonetani(2019)]{nishio2019client}
Takayuki Nishio and Ryo Yonetani.
\newblock Client selection for federated learning with heterogeneous resources in mobile edge.
\newblock In \emph{IEEE International Conference on Communications}, 2019.

\bibitem[Reddi et~al.(2020)Reddi, Charles, Zaheer, Garrett, Rush, Kone{\v{c}}n{\`y}, Kumar, and McMahan]{reddi2020adaptive}
Sashank Reddi, Zachary Charles, Manzil Zaheer, Zachary Garrett, Keith Rush, Jakub Kone{\v{c}}n{\`y}, Sanjiv Kumar, and H~Brendan McMahan.
\newblock Adaptive federated optimization.
\newblock \emph{arXiv preprint arXiv:2003.00295}, 2020.

\bibitem[Shamir(1979)]{shamir1979share}
Adi Shamir.
\newblock How to share a secret.
\newblock \emph{Communications of the ACM}, 1979.

\bibitem[So et~al.(2021{\natexlab{a}})So, Ali, G{\"u}ler, and Avestimehr]{so2021secure}
Jinhyun So, Ramy~E Ali, Ba{\c{s}}ak G{\"u}ler, and A~Salman Avestimehr.
\newblock Secure aggregation for buffered asynchronous federated learning.
\newblock \emph{arXiv preprint arXiv:2110.02177}, 2021{\natexlab{a}}.

\bibitem[So et~al.(2021{\natexlab{b}})So, G{\"u}ler, and Avestimehr]{so2021turbo}
Jinhyun So, Ba{\c{s}}ak G{\"u}ler, and A~Salman Avestimehr.
\newblock Turbo-aggregate: Breaking the quadratic aggregation barrier in secure federated learning.
\newblock \emph{IEEE Journal on Selected Areas in Information Theory}, 2021{\natexlab{b}}.

\bibitem[So et~al.(2022)So, He, Yang, Li, Yu, E~Ali, Guler, and Avestimehr]{so2022lightsecagg}
Jinhyun So, Chaoyang He, Chien-Sheng Yang, Songze Li, Qian Yu, Ramy E~Ali, Basak Guler, and Salman Avestimehr.
\newblock Lightsecagg: a lightweight and versatile design for secure aggregation in federated learning.
\newblock \emph{Machine Learning and Systems}, 2022.

\bibitem[Xie et~al.(2019)Xie, Koyejo, and Gupta]{xie2019asynchronous}
Cong Xie, Sanmi Koyejo, and Indranil Gupta.
\newblock Asynchronous federated optimization.
\newblock \emph{arXiv preprint arXiv:1903.03934}, 2019.

\bibitem[Xu et~al.(2023)Xu, Qu, Xiang, and Gao]{xu2023asynchronous}
Chenhao Xu, Youyang Qu, Yong Xiang, and Longxiang Gao.
\newblock Asynchronous federated learning on heterogeneous devices: A survey.
\newblock \emph{Computer Science Review}, 2023.

\bibitem[Zhang et~al.(2020{\natexlab{a}})Zhang, Li, Xia, Wang, Yan, and Liu]{zhang2020batchcrypt}
Chengliang Zhang, Suyi Li, Junzhe Xia, Wei Wang, Feng Yan, and Yang Liu.
\newblock Batchcrypt: Efficient homomorphic encryption for cross-silo federated learning.
\newblock In \emph{USENIX Annual Technical Conference}, 2020{\natexlab{a}}.

\bibitem[Zhang et~al.(2020{\natexlab{b}})Zhang, Deng, Xu, Sun, Li, and Zheng]{zhang2020attribute}
Yinghui Zhang, Robert~H Deng, Shengmin Xu, Jianfei Sun, Qi~Li, and Dong Zheng.
\newblock Attribute-based encryption for cloud computing access control: A survey.
\newblock \emph{ACM Computing Surveys}, 2020{\natexlab{b}}.

\bibitem[Zhu et~al.(2019)Zhu, Liu, and Han]{zhu2019deep}
Ligeng Zhu, Zhijian Liu, and Song Han.
\newblock Deep leakage from gradients.
\newblock In \emph{Advances in Neural Information Processing Systems}, 2019.

\end{thebibliography}
\bibliographystyle{plainnat}
\newpage
\appendix

\section*{Appendix}

\section{AFL with BASA}
\label{BASA-AFL}
\begin{algorithm}[h]
	\caption{AFL with BASA }
	\label{alg:basa_asyncfl}
	\begin{algorithmic}[1]
		\REQUIRE buffer size $K$, global learning rate $\eta_g$;
		\ENSURE trained global model parameters;
		\STATE Initialize $k \leftarrow 0, t \leftarrow 0, \Bar{\Delta} \leftarrow 0, \mathcal{A} \leftarrow \{a_k^t \mid k<K\} $;
		\FOR{$ i = 0 $ to $ K$}
		\STATE $\mathcal{C}_i \leftarrow \emptyset$;
		\ENDFOR
		\REPEAT
        \STATE Send $w^{t}$ to available users asynchronously;
		\IF{$\text{receive connection from user $u$ }$}
		\STATE Send $\mathcal{A}$, $k$, $\mathcal{C}_k$ to user $u$ and wait for reply;
        \IF{timeout}
            \STATE Abort current connection;
        \ENDIF
		\STATE $\tilde \Delta_k \leftarrow $ masked local updates from user $u$;
		\STATE $\{c_{kj} \mid k<j<K\} \leftarrow $ ciphertexts from user $u$;
		\FOR{$ j = k+1 $ to $ K$}
		\STATE $\mathcal{C}_j \leftarrow \mathcal{C}_j \cup \{c_{kj}\} $;
		\ENDFOR 
		\STATE $\Bar{\Delta} \leftarrow \Bar{\Delta} + \tilde \Delta_k$; 
        \STATE $k \leftarrow k+1$;
		\ENDIF
		\IF{$k ==K$}
		\STATE $\Bar{\Delta} \leftarrow \frac{\Bar{\Delta}}{K}$;
		\STATE $\mathbf{w}^{t+1} \leftarrow \mathbf{w}^{t} - \eta \Bar{\Delta} $;
		\STATE $t \leftarrow t+1, \Bar{\Delta} \leftarrow 0, k \leftarrow 0$;
		\STATE Regenerate attributes set $\mathcal{A} \leftarrow \{a_k^t \mid k<K\} $;
		\FOR{$ i = 0 $ to $ K $}
		\STATE $ \mathcal{C}_i =\emptyset$;
		\ENDFOR
		\ENDIF
		\UNTIL{Convergence}
	\end{algorithmic}
\end{algorithm}

\section{Experiment Details}
\label{appendix-Experiment Details}
\subsection{Implementation}
We implement the training experiments in the FedML framework based on Pytorch \citep{he2020fedml}. All users' training tasks are evenly distributed among 8 NVIDIA Geforce RTX 2080Ti GPUs. The communication between the server and users is simulated through MPI. Our ABE implementation takes the basic algorithm of CP-ABE \citep{bethencourt2007ciphertext}. For large-scale simulation experiments, we use the FLSim framework, which is also used in FedBuff\citep{nguyen2022federated}. 

For the SFL training experiments, we specify the number of concurrently training users in each round. We use FedAvg with momentum as the training algorithm. For the AFL training experiments, we assume the users connect to the server at a constant rate and set a maximum value of concurrently training users. We use FedBuff as the training algorithm. To simulate device heterogeneity, we use an exponential distribution to sample the training delay of devices.

\subsection{Dataset Partition}
For all the datasets used in the experiments, we divided them into different numbers of users in a non-i.i.d manner. Specifically, we partition the MNIST dataset and CIFAR-10 dataset according to the Dirichlet distribution. The parameter $\alpha$ is set to 0.5. This can be configured in the FedML framework \citep{he2020fedml}. The total number of users is 1000 and 100, respectively. For the FEMNIST dataset, we follow the underlying distribution of the raw data, which varies among 3400 users \citep{caldas2018leaf}.

For the large-scale training simulation experiment, we increase the total number of users to 200 when training on the CIFAR-10 dataset.

\subsection{Model Structure}
We train different models for each dataset. For the MNIST dataset, we use a simple logistic regression model, where the dimension is 7850. For the FEMNIST dataset, we train a CNN model recommended in \citep{reddi2020adaptive} for EMNIST experiments. The model has 2 convolutional layers (with 3x3 kernels), max pooling and dropout, followed by a 128-unit dense layer. 

For the CIFAR-10 dataset, we use a CNN model with 2 convolutional layers and 3 fully connected layers for the training experiments. The convolutional layer has 64 channels and 5x5 kernels, with a stride of 1, and padding of 2. This model is integrated into the FedML framework \citep{he2020fedml}. For the simulation training experiment, we train another simple CNN model with 4 convolutional layers and 1 fully connected layer, as the same in \citep{nguyen2022federated}. The convolutional layer has 32 channels and 3x3 kernels, with a stride of 1, and padding of 2.

\section{Additional Experiments}
\label{Additional_Experiments}
\subsection{Impact of Concurrency}
\begin{figure}[t]
	\centering
             \includegraphics[width=0.6\linewidth]{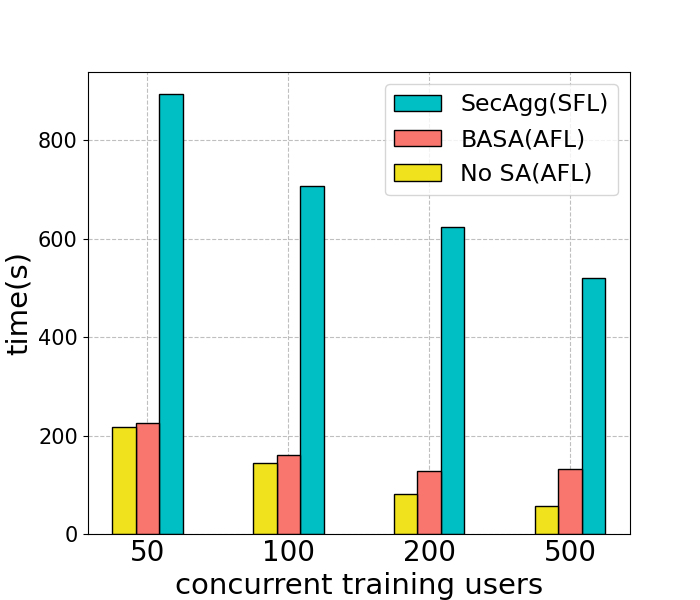}

     \caption{The wall-clock time (in second) required to reach a target validation accuracy for training CNN models on the FEMNIST dataset with different numbers of concurrent training users. }
    \label{concurrency-BASA}
\end{figure}
In the simulation experiments, we also perform the training tasks on the FEMNIST dataset and the CIFAR-10 dataset with different numbers of concurrent users. In Figure \ref{concurrency-BASA}, we compare the total wall-clock time to reach the target accuracy of AFL with BASA, FedBuff(AFL) without secure aggregation, and SFL with SecAgg. The buffer size of BASA and FedBuff is 10, and the delay scale is fixed to 5. As the number of concurrent training users increases, AFL with BASA is consistently more efficient than secure aggregation in SFL with SecAgg.

\end{document}